\documentclass[letterpaper, 10 pt, conference]{ieeeconf}  

\IEEEoverridecommandlockouts                              

\usepackage{times} 
\usepackage{newtxtext}

\usepackage[hyphens]{url}  
\usepackage{graphicx} 
\usepackage{float}
\usepackage[ruled,vlined,linesnumbered]{algorithm2e}
\SetCommentSty{emph}
\usepackage{amsmath,amssymb}

\usepackage{cleveref}

\SetKwInput{KwParam}{Parameter}
\SetKwBlock{Loop}{Loop}{end}

\usepackage{newfloat}
\usepackage{listings}
\floatstyle{ruled}
\newfloat{listing}{tb}{lst}{}
\floatname{listing}{Listing}

\newtheorem{theorem}{Theorem}
\newtheorem{prop}{Proposition}
\newtheorem{lemma}{Lemma}

\usepackage{booktabs}

\title{\LARGE \bf
Scalable Long-Horizon Planning with Staggered Updates for Lifelong MAPF
}

\graphicspath{/Figures}

\author{Vaibhav Sanjay$^{1}$ and Jiaoyang Li$^{1}$
\thanks{$^{1}$Vaibhav Sanjay and Jiaoyang Li are with the Robotics Institute, 
        Carnegie Mellon University.
        {\tt\small \{vsanjay,jiaoyanl\}@andrew.cmu.edu}}%
}

\begin{document}

\maketitle
\thispagestyle{empty}
\pagestyle{empty}

\begin{abstract}
Lifelong Multi-Agent Path Finding (LMAPF) requires generating collision-free paths for large agent fleets under strict real-time constraints.
Reactive frameworks such as PIBT and Enhanced PIBT (EPIBT) scale effortlessly to thousands of agents through rule-based, step-by-step coordination but suffer from severe temporal myopia, making them ineffective in scenarios where long-horizon reasoning is essential. RHCR plans windowed paths over multi-step horizons but incurs substantial planning overheads that hinder scalability. TP tackles both challenges by planning only subsets of agents at each timestep, yet its applicability is restricted to highly structured maps. 
To achieve long-horizon planning at scale across general maps, we propose Path Updates over Staggered Horizons (PUSH), a LMAPF planner capable of coordinating thousands of agents in under a second while planning over multi-step horizons. PUSH combines the key advantages of PIBT, RHCR, and TP. Like TP, PUSH reduces computational complexity by planning only a subset of agents at each timestep using staggered planning windows. Unlike TP, however, PUSH plans RHCR-style windowed paths in general maps without relying on restrictive map assumptions. To maintain high throughput in congested environments, PUSH further integrates EPIBT-inspired priority inheritance, backtracking, and anytime improvements into its windowed planning. Empirical evaluations across two realistic MAPF scenarios requiring long-horizon reasoning show that PUSH scales to the same massive agent loads as EPIBT (e.g., 10k agents) while achieving significantly higher system throughput than all baselines.
\end{abstract}

\section{Introduction}
Multi-Agent Path Finding (MAPF) is the problem of finding collision-free paths for a group of agents, navigating them from their corresponding start locations to their goal locations. MAPF is applicable to several real-world problems, such as planning for robotic sortation centers and autonomous warehouses. To better represent the constant flow of tasks that autonomous warehouses typically experience, researchers typically study a variant of the MAPF problem called Lifelong MAPF (LMAPF). In this variant, agents receive new goals upon reaching their current ones, necessitating a continuous planning cycle for agents.

One of the earliest LMAPF algorithms is Token Passing (TP) \cite{tp}, a long-horizon planner that achieves high scalability by planning paths for agents as they receive new goals in a prioritized planning manner. However, TP does not apply to all map structures in general, limiting its applicability in practice. Rolling Horizon Collision Resolution (RHCR) \cite{rhcr} removes this limitation by planning paths over a fixed, long-horizon window using standard MAPF algorithms and has become the most widely adopted long-horizon LMAPF planner. However, it suffers from limited scalability since it replans all agents simultaneously within each planning window. More recently, Priority Inheritance with Backtracking (PIBT) \cite{pibt} and its enhanced variant EPIBT \cite{epibt} have emerged as highly scalable LMAPF planners for general maps. They replace long-horizon planning with reactive, rule-based coordination through recursive priority inheritance and backtracking. Nevertheless, the price of this computational efficiency is severe temporal myopia.

While (E)PIBT performs impressively under standard LMAPF benchmarks \cite{benchmark}, many real-world applications require long-horizon reasoning that these benchmarks fail to capture. For example, LMAPF assumes that agents can immediately move towards their next goals upon reaching their current ones. In practice, agents often remain stationary while executing tasks such as loading or unloading, creating temporary goal occupancies that induce severe queuing if not anticipated in advance. Furthermore, some industrial layouts feature long, narrow corridors \cite{symbotic, delivery}, requiring agents to anticipate downstream traffic prior to entry. We evaluate against these two representative benchmark test cases in our experiments and discuss additional real-world applications in \Cref{sec:conclusion}.

In this paper, we introduce Path Updates over Staggered Horizons (PUSH), an LMAPF algorithm that maintains the robust, long-horizon reasoning of frameworks like RHCR while matching the extreme scalability and congestion-handling capabilities of reactive planners like EPIBT. PUSH uniquely bridges the gap between subset-agent planners like TP, which scale by routing a subset of agents over full horizons, and windowed planners like RHCR, which plan windowed horizons over the entire fleet. Our framework maintains a global set of safe fallback paths, enabling selective, windowed replanning for only a partial subset of agents across unaligned temporal windows. Furthermore, PUSH includes EPIBT-inspired recursive displacement. Agents calculate fleet-wide priorities before planning, allowing higher-priority agents to actively push lower-priority agents and trigger recursive replanning. This allows for highly scalable planning even in severely congested scenarios. Finally, we integrate a localized, anytime Large Neighborhood Search routine tailored to our staggered windowed structure to continually drive plans toward optimality within strict real-time execution budgets. We compare the algorithms on two benchmarks that require long-horizon reasoning and demonstrate that our approach outperforms all baselines significantly in system throughput.

\section{Related Work}
We provide an introduction to three different LMAPF planning paradigms, along with an introduction to an example in each category. We then discuss anytime optimization techniques commonly used in online planning methods.

\subsection{Subset Planning}
One idea for scaling LMAPF is to avoid planning all agents at the same time. Instead, only agents that have completed their current paths are replanned.
Existing algorithms in this category~\cite{tp, ta-prioritized, sbda} rely on specialized map structures and/or goal assignment algorithms. PUSH, to our knowledge, is the first algorithm in this category that requires neither. We use TP~\cite{tp} as a representative example to illustrate the key ideas behind this class of algorithms.

\subsubsection{TP}
Token Passing (TP) is an early LMAPF algorithm that replans only agents that just received new goals. Although replanning only this subset drastically reduces planning complexity, TP treats all other agents as higher-priority agents, preventing newly replanned agents from colliding with any of them. This rigidness frequently leads to unsolvable planning queries and permanent deadlocks whenever an already-planned path obstructs a critical bottleneck. Thus, TP's solvability guarantees are fundamentally restricted to \emph{well-formed} maps, which ensure that every pair of goals is connected by a path that avoids all other goals.

\subsection{Windowed Planning}
Another prevalent idea for scaling LMAPF is windowed planning, which bounds the temporal complexity of the problem by optimizing agent paths over a fixed number of timesteps. Examples include WHCA* \cite{whca} and the more general framework RHCR \cite{rhcr}, which we describe below. 

\subsubsection{RHCR}
Rolling Horizon Collision Resolution (RHCR) divides the lifelong problem into fixed-length episodes. It is governed by two key hyperparameters: the planning horizon $W$ and the execution horizon $K$ (where $K \leq W$). At the beginning of each episode, RHCR uses a windowed MAPF planner to generate $W$-step collision-free paths for the entire fleet. The agents then execute the first $K$ steps before replanning. RHCR works with many MAPF planners on general maps, making it a popular framework with many follow-up works~\cite{ggo,symbotic}.
Subsequent extensions to RHCR enhance reasoning in long corridors \cite{rhcr-corridor}, employ learning-based agent prioritization \cite{symbotic}, reuse experience from previous windowed horizons \cite{exrhcr}, or dynamically adjust graph edge weights to reduce congestion \cite{ggo}. However, despite these refinements, RHCR remains fundamentally bottlenecked at scale because replanning the entire fleet simultaneously incurs costly overheads in large fleets or dense traffic.

\subsection{Reactive Planning}
Reactive planning encompasses a wide variety of scalable approaches that optimize agent movements over short temporal horizons, typically considering only the next step. Representative examples include rule-based, PIBT-style algorithms derived from Priority Inheritance with Backtracking \cite{pibt, epibt} and learning-based algorithms \cite{scrimp, primal2}. Recently, hybrid algorithms merge these paradigms to deliver even stronger performance \cite{sillm}. In this paper, we focus on PIBT-style algorithms, which serve as the foundational bedrock for state-of-the-art reactive planners.

\subsubsection{PIBT}
Priority Inheritance with Backtracking (PIBT) \cite{pibt} utilizes a reactive, one-step decision rule. At each timestep, PIBT establishes a global priority ordering across the agent fleet. Agents then plan sequentially in descending order of priority, each selecting its best next action. If a higher-priority agent chooses a move that pushes a lower-priority agent, the displaced agent immediately plans its own evasive move, temporarily inheriting the higher priority. If an agent fails to find a valid action, PIBT backtracks along the dependency chain, forcing the preceding agent to select its next best action. Clearly, the priority ordering is crucial to its performance. Common priorities include longer elapsed time (sorting agents by elapsed timesteps since reaching the last goal) and closer goal (sorting agents by distance to the next goal)~\cite{Jiang2026md-pibt}.
Subsequent papers improved the next-step selection procedure to better avoid congestion, improving the throughput of the system~\cite{traffic-flow}. Other papers improved upon the algorithm itself, such as using regret learning to better improve tiebreaking action selections \cite{pibt-regret}. Papers such as LaCAM \cite{lacam} utilize PIBT as a fast low-level successor generator to efficiently explore many joint-state agent configurations.

\subsubsection{EPIBT} To bridge the gap between the strict one-step nature of PIBT and longer-horizon windowed methods, Enhanced PIBT (EPIBT) \cite{epibt} extends PIBT with multi-action planning by enabling agents to reserve paths of up to $W$ actions. However, due to its rule-based nature, EPIBT keeps $W$ small (i.e., 3-5) and restricts each planned path to push at most one lower-priority agent. It follows the same recursive priority inheritance and backtracking framework as PIBT to push conflicting agents out of the way.

Despite their fast planning speed, reactive planners generally lack the long-horizon foresight needed to handle certain real-world scenarios due to their limited planning horizons.

\subsection{Anytime Optimization}
Anytime optimization improves solution quality during the idle time between the completion of planning and the per-timestep runtime limit. It is particularly effective for online planning, allowing an initial feasible solution to be refined without delaying execution, and is therefore used by many LMAPF systems \cite{epibt,pie,wppl,flatland}.

\subsubsection{LNS} A common approach is Large Neighborhood Search (LNS). Given an existing MAPF solution, MAPF-LNS \cite{lns} repeatedly selects a subset of agent paths, removes them, and replans using Prioritized Planning (PP) while treating unselected paths as obstacles.
In highly congested scenarios where PP often fail to find a feasible repair, EPIBT proposes a modified LNS tailored to its solution. It repeatedly selects a random agent, removes its $W$-step path, and assigns it the highest priority to override neighboring paths. If the resulting solution improves upon the current one, it is accepted. This iterative process runs until the timestep budget is exhausted.

\section{Problem Formulation}
\label{sec:prob_formulation}
LMAPF is defined on a graph, 
typically a grid map where adjacent cells are connected by edges. We want to navigate a set of $N$ agents, $\mathcal{A} = \{r_1,...,r_N \}$, from their current positions to their assigned goals. Upon completing a task at its current goal, an agent is assigned a new goal by a given scheduler. Time is discretized into timesteps, and at each timestep an agent may either move to an adjacent vertex or wait in place. A collision occurs when two agents occupy the same vertex at the same timestep (vertex collision) or swap vertices in consecutive timesteps (edge collision). Our goal is to maximize throughput (the average number of goals reached per timestep) while avoiding collisions.

\subsubsection{Task Completion Time}
To simulate the effect of agents completing tasks at their goals (e.g., loading/unloading packages), we introduce the \emph{Task Completion Time} (TCT).
If the TCT is $T$ timesteps, an agent reaching its goal at timestep $t$ must remain there until timestep $t + T$ before receiving a new goal from the scheduler. We further assume that TCT may be stochastic and is unknown in advance. When $T = 0$, this formulation reduces to standard LMAPF.

\section{Method}
We introduce Path Updates over Staggered Horizons (PUSH), an LMAPF algorithm that combines key principles from TP, RHCR, and EPIBT to simultaneously achieve long-horizon planning, high scalability, and general map applicability.

Conceptually, PUSH extends Token Passing (TP), which scales by replanning only agent subsets and maintaining a set of collision-free fallback paths $\Pi$. To ensure that a new path always exists, TP maintains a set of collision-free paths $\Pi$. Importantly, TP requires $\Pi$ to be free of target collisions; that is, if agents reaching their goals wait there indefinitely, they will not collide with any other agents. TP replans an agent when it reaches its current goal. In the worst case, the agent can therefore wait at its current goal until all other agents reach their goals before moving to its next goal, guaranteed by the assumption of a well-formed map. PUSH eliminates this constraint by integrating RHCR’s rolling-horizon lookahead: rather than planning fully to the goal, agents replan $W$-step paths that remain collision-free against $\Pi$. A feasible path always exists because any agent can always default to its existing fallback path in $\Pi$.

While this baseline (denoted as PUSH-lite in our experiments) yields an efficient planner, forcing replanning agents to treat all other agent's paths as immutable hard constraints leads to excessive waiting. To resolve this, PUSH maintains persistent agent priorities as in PIBT and EPIBT and allows higher-priority agents to recursively displace lower-priority neighbors, drastically improving path quality while preserving scalability and map generality. 

We now formally describe PUSH. 

\subsection{High-Level Planning}
PUSH relies on three predefined parameters: $W$, $K$ ($K \leq W$), and $M$. As in RHCR, $W$ and $K$ denote the planning and execution horizons, respectively. Borrowing from EPIBT, $M$ bounds the maximum number of times an agent can be revisited during one recursive search, preventing excessive computation time.

PUSH calls \Cref{alg:push} at the start of every timestep before agents move. The input includes the set of agents $r_i \in \mathcal{A}$, along with their current positions $s_i$, goals $g_i$, and priorities $p_i$. 
PUSH also maintains a countdown array $\textit{time2Replan}$, indicating when each agent should be replanned, and a set of $W$-step collision-free paths $\Pi$ for all agents. Before timestep 0, $\textit{time2Replan}$ is set to evenly distributed values between $0$ and $K - 1$ inclusive, which facilitates PUSH's subset planning behavior. $\Pi$ is initialized with $W$-step waiting-in-place operations for each agent. 

At the beginning of each timestep, 
PUSH takes $\Pi$ from the previous timestep, removes the first action (which has already been executed) from every path, and appends a \textsc{wait} action (line \ref{line:do_action}). This immeiately gives us a set of $W$-step collision-free paths. However, these paths contain many \textsc{wait} actions and can likely be improved. Therefore,
PUSH then identifies the subset $Q \subseteq A$ of agents that require path updates. We first skip task executing agents (lines \ref{line:task_skip_start}). Agents that have reached their goals and are assigned \textit{no-op} paths (i.e., waiting at their goals for the entire horizon) begin executing their tasks (lines \ref{line:noop_check}-\ref{line:begin_task}). For all other agents, PUSH decrements $\textit{time2Replan}[r_i]$ (line \ref{line:decrement}), then checks if the agent has exhausted its execution window or its next action is a wait in place (line \ref{line:check_t2r}). If so, it is inserted into $Q$ (line \ref{line:add_queue}). $Q$ is then sorted in order of priority, breaking ties by agent ID (line \ref{line:sort_priority}).

Next, PUSH replans each agent in $Q$ in order. It first constructs a protected set of agents $\mathcal{H}$, consisting of higher-priority agents together with agents currently executing tasks ($\mathcal{H}_0$),  which cannot be pushed during replanning (line \ref{line:rec_stack_set}). 
It then initializes the global visit counter \textit{visitCnt}, which limits recursive replanning to at most $M$ vists per agent (line \ref{line:visit_count_set}). Finally, PUSH invokes $\textsc{PushPlan}$ to compute a new $W$-step path through recursive priority-pushing (\cref{line:push_plan_call}).

\subsection{Priority Inheritance Pushing}
\begin{algorithm}[tb]
\caption{PUSH}
\label{alg:push}
\KwIn{Agents $r_i \in \mathcal{A}$ with positions $s_i$, goals $g_i$, and priorities $p_i$, their time to replan $\textit{time2Replan}$ and their paths $\Pi$ from previous timestep}
\KwOut{Updated paths $\Pi$ for the current timestep}
\ForEach{path $\pi \in \Pi$}{
    Remove $\pi[0]$ from $\pi$ and append a \textsc{wait} action\; \label{line:do_action}
}
$Q \gets \emptyset$\tcp*{Agents that need to be replanned}\label{line:queue}
$\mathcal{H}_0 \gets \emptyset$\tcp*{Agents that are executing tasks}
\ForEach{agent $r_i \in \mathcal{A}$}{
    \lIf{$r_i$ is doing task}{\label{line:task_skip_start}$\mathcal{H}_0 \gets \mathcal{H}_0  \cup \{r_i\}$\label{line:task_skip_end}}
    \uElseIf{$\Pi[r_i]$ is a no-op path and $s_i = g_i$}{ \label{line:noop_check}
        $r_i$ starts its task\;\label{line:begin_task}
        $\mathcal{H}_0 \gets \mathcal{H}_0  \cup \{r_i\}$\;
    }
    \Else
    {
        $\textit{time2Replan}[r_i] \gets \textit{time2Replan}[r_i] - 1$\;\label{line:decrement}
        \lIf{$\textit{time2Replan}[r_i] = 0$ or $\Pi[r_i][0] = \textsc{wait}$ \label{line:check_t2r}} {$Q \gets Q \cup \{r_i\}$}\label{line:add_queue}
    }
}
Sort $Q$ by priority, 
breaking ties by agent ID\;\label{line:sort_priority}
\ForEach{agent $r_i \in Q$ \label{line:for-loop}}{
    $\mathcal{H} \gets \mathcal{H}_0  \cup \{r_j \in A \mid p_{j} \textnormal{ is higher than } p_i\}$\;
    \label{line:rec_stack_set}
    $\textit{visitCnt}[r_i] \gets 0, \forall r_i \in \mathcal{A}$\; \label{line:visit_count_set}
    $\textsc{PushPlan}(r_i, \mathcal{H}, \Pi, \textit{visitCnt})$\;\label{line:push_plan_call}
}
\Return $\Pi$\;
\end{algorithm}

\begin{algorithm}[tb]
\caption{\textsc{PushPlan} in PUSH}
\label{alg:push-plan}
\KwIn{Agent $r_i$, set of higher-priority agents $\mathcal{H}$, paths for all agents $\Pi$, and visit counts \textit{visitCnt}}
\KwOut{\textsc{true} if a collision-free path for $r_i$ was found, \textsc{false} otherwise}

\lIf{$\textit{visitCnt}[r_i] \geq M$}{\label{alg:max_visit}\Return \textsc{false}}
$\textit{visitCnt}[r_i] \gets \textit{visitCnt}[r_i] + 1$\;\label{alg:visit_count_inc}
$\textit{cons} \gets \emptyset$\tcp*{Constraints}
$\pi^{\mathrm{old}} \gets \Pi[r_i]$\;\label{alg:save_old}
\Loop{\label{line:loop}
    $(\pi, r_d, t_d) \gets \textsc{PlanPath}(r_i, \textit{cons},  \mathcal{H},  \Pi)$\; \label{alg:call_pap}
    \If{$\pi = \emptyset$}{\label{alg:no_sol}
        $\Pi[r_i] \gets \pi^{\mathrm{old}}$\;\label{alg:rec_fail_start}
        \Return \textsc{false}\;
    }
    $\Pi[r_i] \gets \pi$\;
    \If{$r_d$ does not exist or $\textsc{PushPlan}(r_d, \mathcal{H} \cup \{r_i\}, \Pi,  \textit{visitCnt})$}{\label{alg:no_col}
        $\textit{time2Replan}[r_i] \gets K$\;\label{line:reset_steps}
        \Return \textsc{true}\; \label{line:ret_true}
    }
    {
        $\textit{cons} \gets \textit{cons} \cup \{(r_d,\; t_d)\}$\;\label{alg:add_avoid_set}
    }
}
\end{algorithm}
\begin{figure}[t]
\centering
\includegraphics[width=0.40\textwidth]{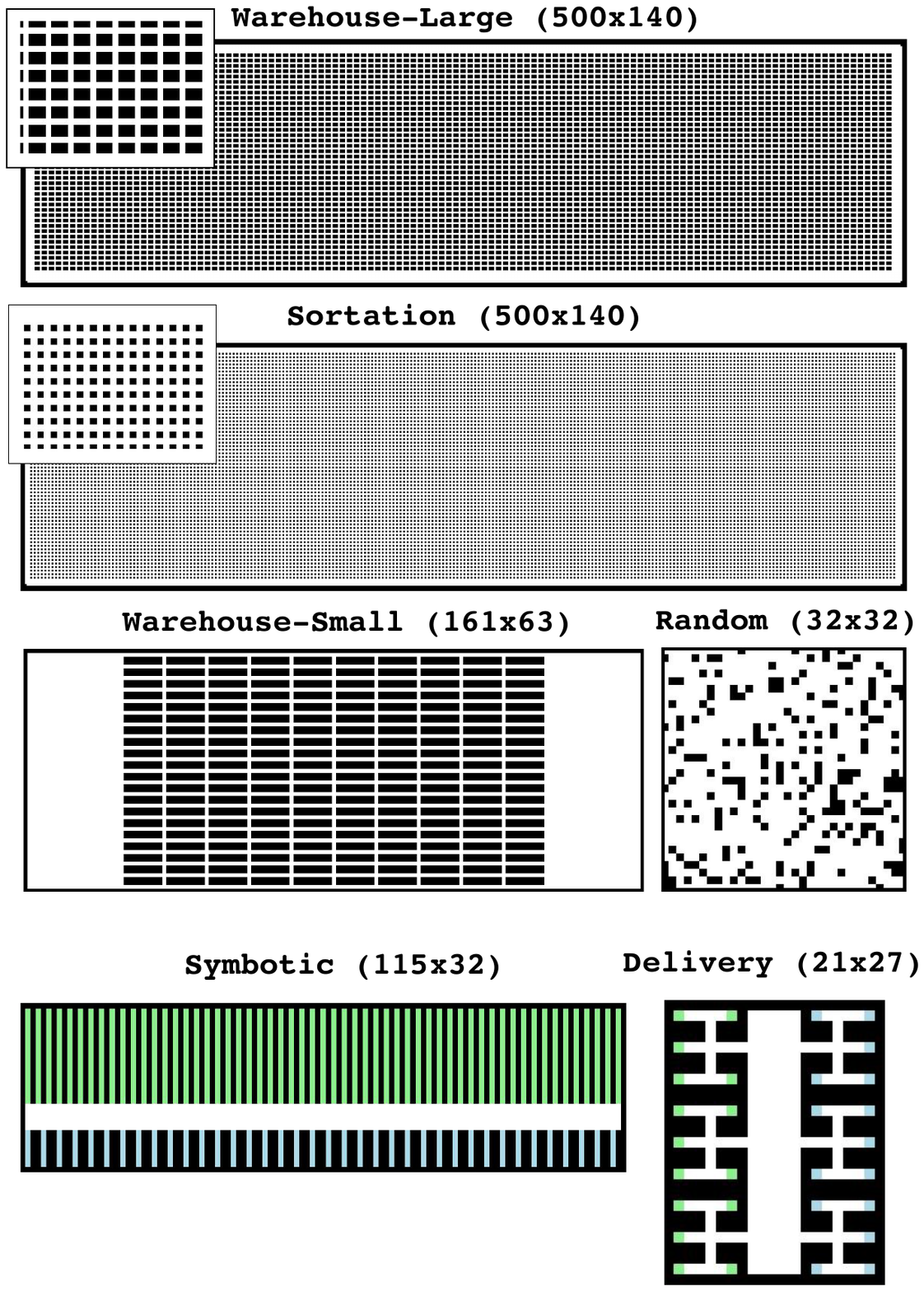} 
\caption{A visualization of the maps used in the experiments. For the Symbotic and Delivery maps, agent goal locations must swap between green and blue tiles. For all other maps, agent start and goal locations are randomly distributed across the free spaces.}
\label{fig:maps}
\end{figure}
Algorithm \ref{alg:push-plan} shows \textsc{PushPlan}, which aims to plan a new $W$-step path for agent $r_i$. 
At a high level, \textsc{PushPlan} follows the priority inheritance and backtracking mechanism of EPIBT. It first plans a path for $r_i$ that collides with at most one lower-priority agent in $\mathcal{A} \setminus \mathcal{H}$ (line \ref{alg:call_pap}). If such a collision occurs, the displaced agent $r_d$ inherits the protected set $\mathcal{H}$ of $r_i$, adds $r_i$ to $\mathcal{H}$, and is recursively replanned (line \ref{alg:no_col}). The recursion continues until either a collision-free solution is found, in which case the search unwinds, and all agents in the dependency chain adopt their new paths and reset  $\textit{time2Replan}$ to $K$ (line \ref{line:reset_steps}), or no feasible paths exists, triggering backtracking (line \ref{alg:no_sol}). 

To support larger planning horizons, PUSH differs from EPIBT in three key aspects. The first difference, already introduced in \Cref{alg:push}, is that PUSH replans only a subset of agents at each timestep and executes each planned path for up to $K$ steps instead of one. 

The second difference lies in the realization of \textsc{PlanPath} (\cref{alg:call_pap}), which computes an optimal path for $r_i$ that satisfies the constraints in $\textit{cons}$ (explained later) and collides with at most one lower-priority agent. Unlike EPIBT, which precomputes all $W$-step paths, PUSH runs a modified windowed space-time A*, making large planning horizons computationally tractable. 
Each search state records the displaced agent (if any), allowing repeated collisions with the same lower-priority agent while prohibiting collisions with multiple agents. The search also enforces the constraints in $\textit{cons}$, pruning states that contain a vertex/edge collision with agent $r_d$ at timestep $t_d$ for $(d, t_d) \in \textit{cons}$.
The \textsc{PlanPath} procedure returns $\pi$, a valid $W$-step path if it exists, $r_d$, an agent that this path collides with, and $t_d$, the timestep of the collision. Since multiple $\pi$ might exist, \textsc{PlanPath} returns the path that minimizes $h(\pi[W], g_i)$, where $h$ is the shortest path distance function and $g_i$ is agent $r_i$'s next goal (i.e. the path that takes $r_i$ closest to its goal).

The last difference occurs during backtracking. 
When no path is found for an agent (line \ref{alg:no_sol}), \textsc{PushPlan} restores its old path in $\Pi$ and returns \textsc{false}, forcing the preceding agent to replan. EPIBT would continue by enumerating another candidate path for the preceding agent. For large $W$, however, exhaustive path enumeration becomes prohibitively expensive. Instead, PUSH records the failed collision as a constraint $(r_d, t_d)$ (line \ref{alg:add_avoid_set}), requiring future A* searches to avoid the vertex or edge collision with agent $d$ at timestep $t_d$. This prevents revisiting the same failure on subsequent loops and provides an effective way to explore alternative paths.

\subsection{Theoretical Analysis}
Like TP and EPIBT, PUSH guarantees to return collision-free paths for all agents at every timestep in finite time. This is because \textsc{PushPlan} on \Cref{alg:push} \cref{line:push_plan_call} always returns \textsc{true}, as it can at least find a path that is identical to $r_i$'s old path in $\Pi$. Moreover, \textsc{PushPlan} always terminates in finite time because (1) the recursion depth of \textsc{PushPlan} is bounded by the number of agents, since each agent is added to the protected set $\mathcal{H}$ whenever it pushes another agent, preventing cyclic recursion; and (2) the number of iterations of the loop (\cref{line:loop}) is bounded by the number of possible collisions an agent can encounter within $W$ timesteps. The following theorem formalizes this result; its proof is provided in the appendix, along with a time complexity analysis.

\begin{theorem}
At every timestep, PUSH guarantees to return $\Pi$ in finite time, and $\Pi$ is guaranteed to be collision-free. 
\end{theorem}

Furthermore, we show that PUSH preserves the same deadlock-freedom guarantee as EPIBT: every agent is guaranteed to reach its goal within a finite number of timesteps if the longer elapsed time is used as the priority ordering, the graph is biconnected, and TCT is 0 (i.e., agents do not perform tasks upon reaching their goals). The following theorem formalizes this result; its proof and a discussion of why the guarantee may not hold for either PUSH or EPIBT when TCT is non-zero are provided in the appendix.

\begin{theorem}
    If the longer elapsed time is used as the priority ordering, the graph is biconnected, and TCT is 0, then PUSH is deadlock-free: every agent reaches its goal within a finite number of timesteps.
\end{theorem}

\subsection{Anytime Optimization}
PUSH takes advantage of EPIBT-style LNS to optimize its solutions during alloted time after \Cref{alg:push} completes and before the next planning cycle must start. Similarly to EPIBT, PUSH selects a random agent $r_i$ that is not doing tasks (i.e., not in $\mathcal{H}_0$) and call \textsc{PushPlan} on this agent. Given the old set of paths $\Pi^{\mathrm{old}}$ and the updated paths $\Pi$, we compare them to see if the replanning improved the solution quality. To do so, we calculate $\sum _{i=1}^Nh(e_i, g_i)$, where $e_i$ is the last position in the path of agent $r_i$. The set with the lesser sum is the one with agents closer to their goals on average, and is thus the set of paths we should accept. We repeat this process until the time limit is reached.

\section{Experimental Results}\label{sec:exp}
In this section, we empirically evaluate PUSH against several baselines. Our experiments are performed on a machine with an AMD Threadripper 3990X 64-core processor with 192 GB of RAM. We base our implementation on the official repository for EPIBT, made in C++.

\subsection{Experiment Setup}
\subsubsection{Testing Scenarios}
\begin{figure*}[t]
\centering
\includegraphics[width=0.95\textwidth]{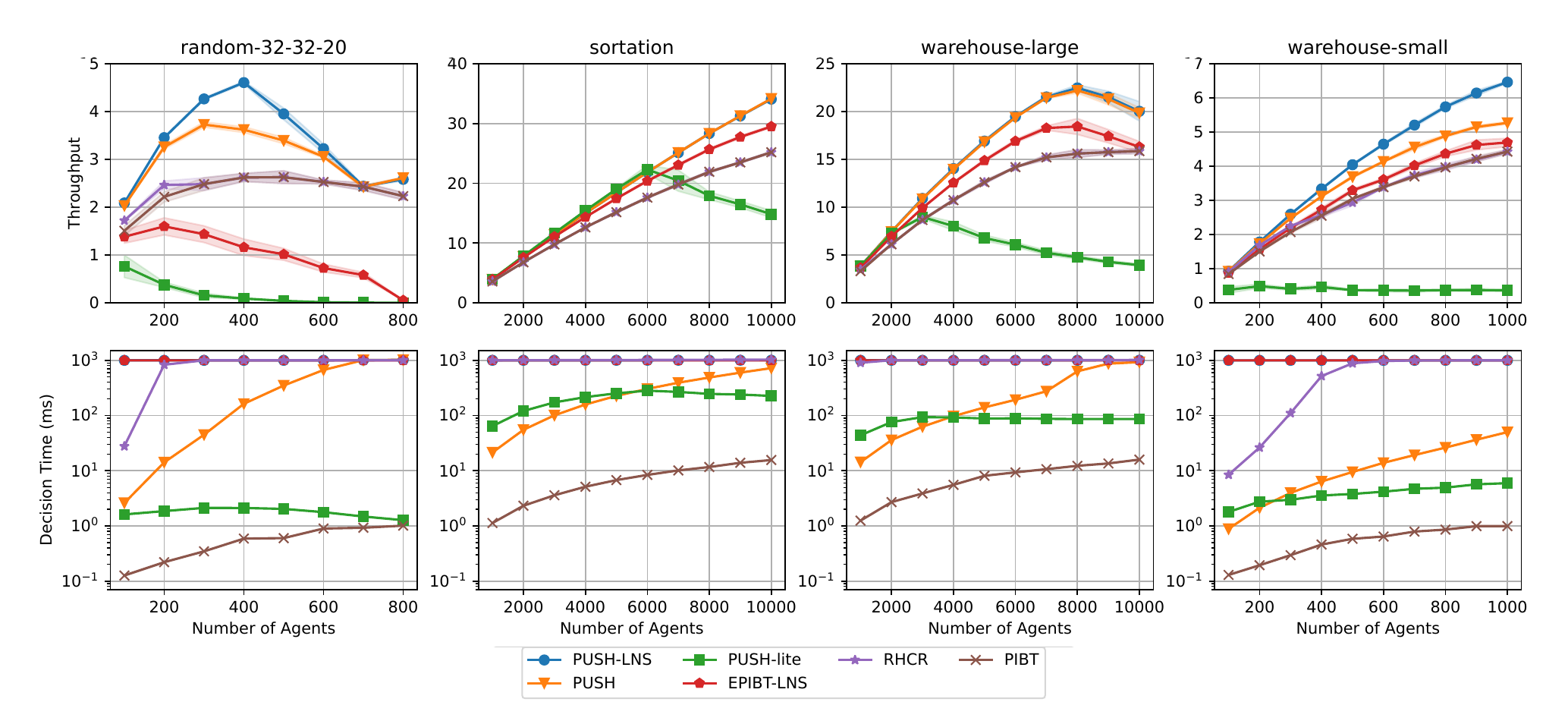} 
\caption{Comparison of PUSH with baselines on TCT maps assuming a task completion time of 20 timesteps. 
}
\label{fig:tct}
\end{figure*}
We evaluate PUSH across two long-horizon scenarios: (1) TCT Maps (\textsc{random-32-32-20}, \textsc{sortation}, \textsc{warehouse-large}, \textsc{warehouse-small}) test the algorithms with large TCTs. These maps can be found in the MAPF benchmark \cite{benchmark} or the League of Robot Runners competition \cite{lorr}. The goal positions are evenly sampled from the free space in the map. (2) Complex Maps (\textsc{symbotic},\footnote{It is a representative warehouse map used by Symbotic.} \textsc{delivery}) feature long corridors and dead-ends with zero TCTs. The \textsc{Symbotic} map \cite{symbotic} contains many vertical parallel and a long horizontal deck. The \textsc{delivery} map \cite{delivery} features many short tree-like pathways with goals deep inside. The goals for each agent are sampled alternately between green and blue tiles. \Cref{fig:maps} shows these maps.

\subsubsection{Baselines}
\begin{figure}[t]
\centering
\includegraphics[width=0.475\textwidth]{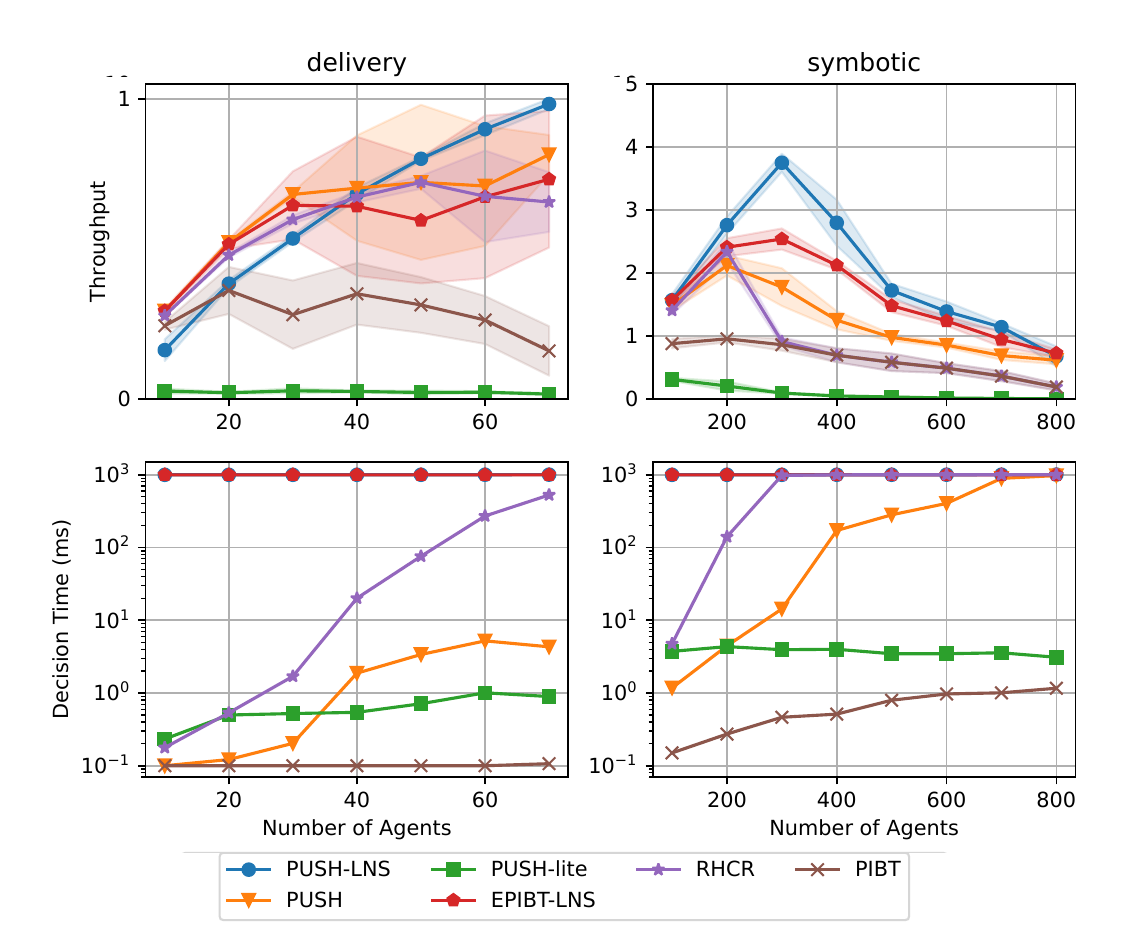} 
\caption{Comparison of PUSH with baselines on complex maps.}
\label{fig:complex}
\end{figure}
\begin{figure*}[t]
\centering
\includegraphics[width=0.95\textwidth]{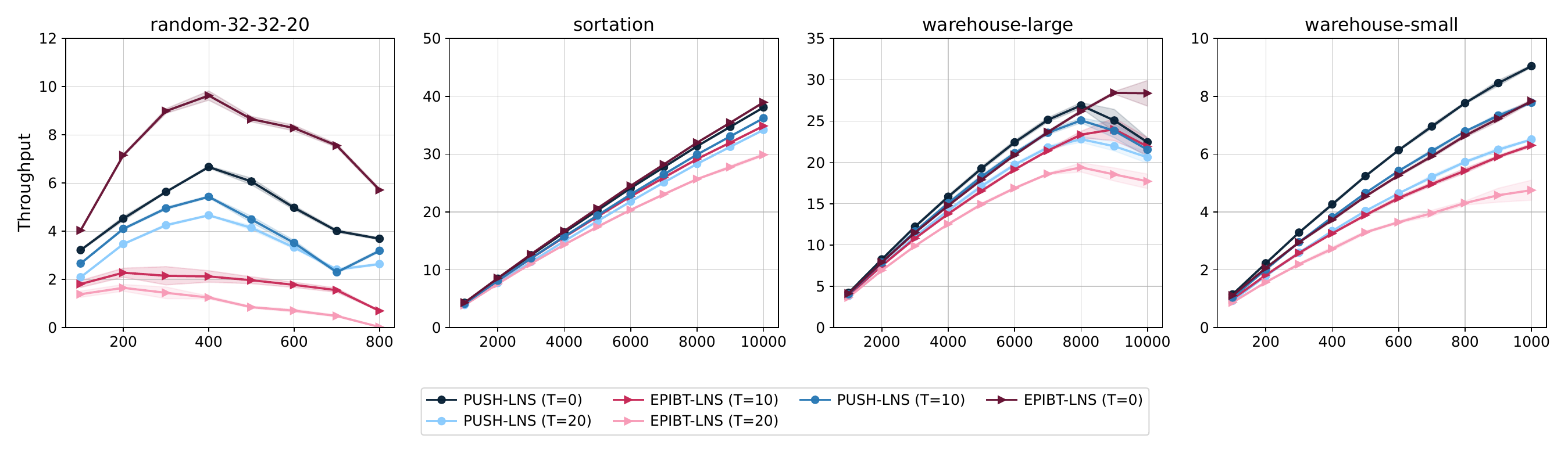} 
\caption{Comparison of PUSH-LNS with EPIBT-LNS with various TCTs. Brighter colors indicate longer TCTs.}
\label{fig:tct_compare}
\end{figure*}
\begin{figure}[t]
\centering
\includegraphics[width=0.475\textwidth]{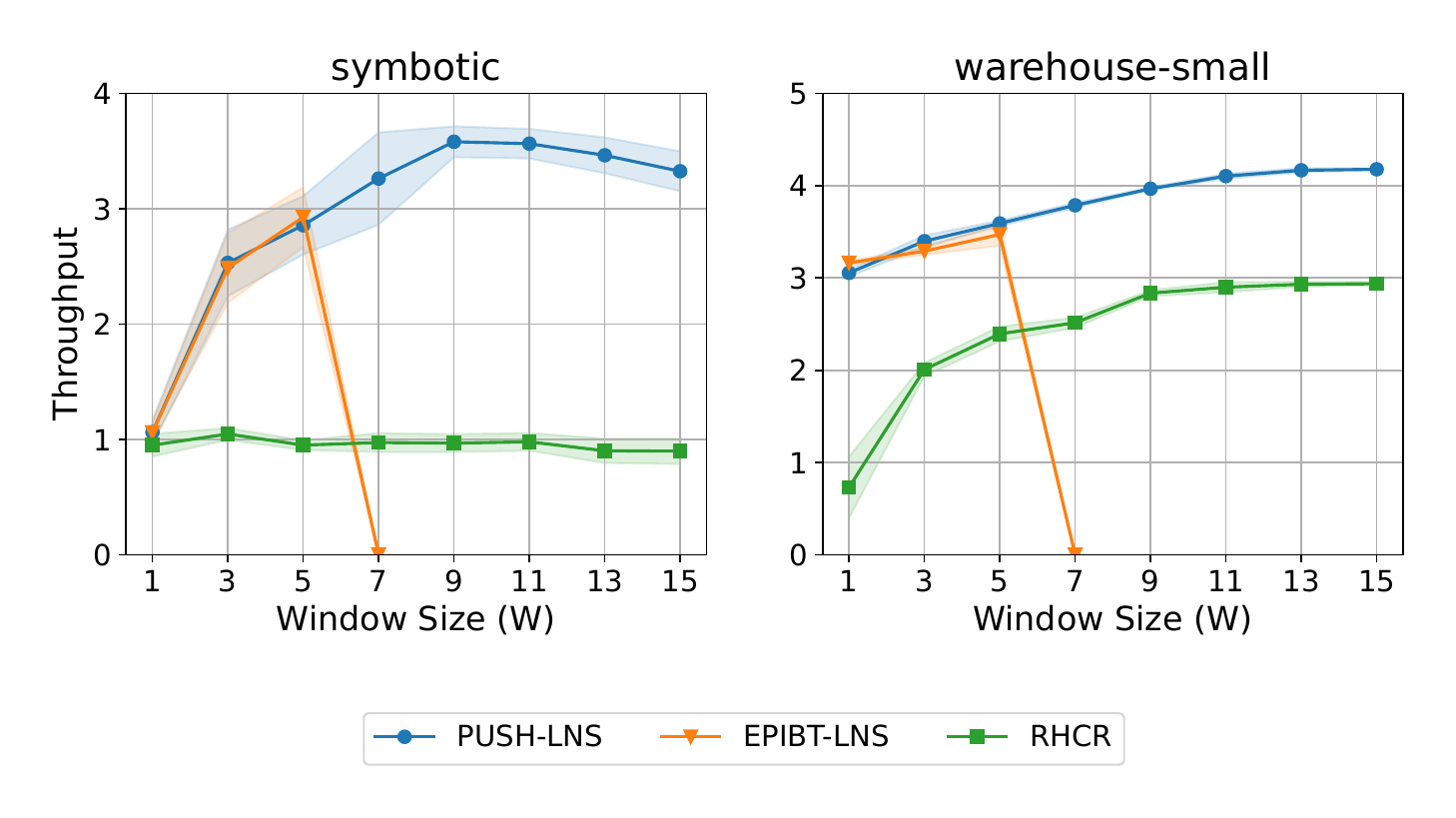} 
\caption{Comparison of PUSH with baselines when varying planning window sizes. We used 300 agents and a TCT of 0 timesteps in \textsc{symbotic}. We used 500 agents and a TCT of 20 timesteps in \textsc{warehouse-small}.}
\label{fig:window_abl}
\end{figure}
We evaluate PUSH against representative baselines from across the main LMAPF paradigms. We exclude subset-planning frameworks (e.g., Token Passing) due to their strict topological map requirements. From windowed planning, we select RHCR with a PBS solver (RHCR-PBS), the current state of the art. From reactive planning, we select the state-of-the-art training-free baselines PIBT and EPIBT with LNS (EPIBT-LNS). Additionally, we include an ablation model, PUSH-lite, which omits recursive priority pushing and instead assigns newly replanned agents the lowest priority. To evaluate the impact of anytime optimization, we test two variants of our method: PUSH and PUSH-LNS. 

\subsubsection{Hyperparameters}
The 3 important hyperparameters $W$, $K$, and $M$ are set based on the algorithm. PUSH, PUSH-LNS, PUSH-lite, and RHCR all use $W = 10$ and $K = 5$ for fair comparison. EPIBT uses $W = 3$ and $K = 1$ (best reported parameters in the original paper), while PIBT uses $W = 1$ and $K = 1$ as it does not support multi-step planning. PUSH, PUSH-LNS, and EPIBT set $M = 10$. Furthermore, the method of priority ordering is important to consider for both PUSH variants as well as PIBT and EPIBT. For the TCT maps, we utilize standard ``close goal" priorities, where agents closer to their goals receive higher priorities. For the complex maps, we switch to ``far goal" priorities where agents farther to their goals receive higher priorities. This intuitively lets agents who just completed a task deep within a dead end escape without being forced inside by agents taking space at the sole exit and empirically performs better.

\subsubsection{Evaluation Criteria}
We run each experiment for 1,000 timesteps with a 1-second per-timestep planning budget and report throughput (goals reached per timestep) averaged over 5 randomized trials, together with its standard deviation (shown as shaded regions in the figures). We also report average decision time per timestep. 
To enforce the planning budget, if any algorithm besides RHCR times out at a given timestep, the algorithm terminates, and unplanned agents fall back on their paths from the previous timestep, which are guaranteed to be collision-free. Since RHCR replans every $K$ timesteps, it receives $K$ seconds of planning time per rolling horizon for fair comparison. If it times out, we default to PIBT for $K$ steps before returning back to RHCR.

\subsection{Results on TCT Maps}
Figure \ref{fig:tct} compares the performance of PUSH and baselines in the TCT maps, where agents remain stationary at goal locations for $T=20$ timesteps upon arrival. Across all evaluated maps and fleet sizes, PUSH and PUSH-LNS achieve the strongest throughput. The results highlight distinct failure modes: while RHCR provides competitive throughput at low agent densities, the planner rapidly exhausts its planning budget as congestion increases, forcing search truncation and a fallback to myopic PIBT execution that degrades throughput. Purely reactive planners such as PIBT and EPIBT-LNS maintain smooth throughput scaling in layouts with abundant short detours like \textsc{sortation}, but struggle in narrow corridors (\textsc{warehouse-small}, \textsc{warehouse-large}) and dense obstacles (\textsc{random-32-32-20}). Furthermore, while the PUSH-lite ablation matches PUSH at low agent counts in large maps, its performance suffers a catastrophic collapse as fleet size scales beyond 3,000 agents in \textsc{warehouse-large} and 6,000 agents in \textsc{sortation}. In contrast, PUSH and PUSH-LNS achieve  impressive maximum throughputs, delivering a 25\% maximum throughput improvement in \textsc{warehouse-large} and a 300\% maximum throughput improvement in \textsc{random-32-32-20} over EPIBT-LNS.

\subsection{Results on Complex Maps}
To isolate the structural impact of maps on algorithmic performance, Figure \ref{fig:complex} evaluates on complex maps with instantaneous task completion ($T = 0$). PUSH-LNS achieves the highest throughput, at around 1.0 in \textsc{delivery} and reaching 3.8 in \textsc{symbotic}, while maintaining tight standard deviations. The experimental results reveal distinct failure modes among the baselines under topological constraints: EPIBT-LNS struggles with high agent densities within deep corridor networks, and its performance is highly sensitive to start-goal configurations, manifesting as high throughput variance (reflected in the wide shaded standard deviations in \textsc{delivery}). As fleet size increases, RHCR is unable to match PUSH-LNS in \textsc{delivery}, and in \textsc{symbotic} it frequently exhausts its 5-second per-horizon budget. Finally, the performance failure of PUSH-lite across both maps illustrates the role of recursive priority displacement in non-well-formed maps.

\subsection{TCT Ablation}
Figure \ref{fig:tct_compare} evaluates the performance sensitivity of PUSH-LNS and EPIBT-LNS to varying Task Completion Times ($T \in \{0, 10, 20\}$), omitting decision time plots as both algorithms fully utilize the 1-second per-timestep planning budget. The severe vulnerability of myopic planners to non-zero TCTs is particularly apparent in the dense \textsc{random-32-32-20} map: while EPIBT-LNS achieves high throughput under standard LMAPF ($T = 0$), its performance degrades when $T = 10$ or $20$. In contrast, PUSH-LNS demonstrates exceptional performance, retaining high throughput outperforming EPIBT-LNS across all non-zero TCTs. While the topological flexibility and abundant alternative pathways of the \textsc{sortation} map partially cushion EPIBT-LNS from degradation, the deep, narrow corridors of the warehouse environments severely penalize reactive planners.

\subsection{Window-Size Ablation}
\Cref{fig:window_abl} illustrates the effect of planning window size on PUSH-LNS, EPIBT-LNS, and RHCR. For the \textsc{symbotic} map, we used 300 agents and a TCT of 0. For the \textsc{warehouse-small} map, we used 500 agents and a TCT of 20. An interesting occurrence is that PUSH-LNS and EPIBT-LNS report very similar throughputs with the same window sizes of 1, 3, and 5. This result is expected because both planners rely on the same recursive backtracking mechanism for planning agent paths. However, EPIBT-LNS falls off dramatically at $W = 7$, due to the exponential explosion in number of possible operations at high $W$. In contrast, PUSH-LNS faces no such restriction, giving it access to long horizon planning. RHCR cannot keep up with PUSH-LNS or EPIBT-LNS at this scale, and beyond $W = 8$ in \textsc{warehouse-small} and $W = 3$ in \textsc{symbotic}, it completely defaults to PIBT planning.

\section{Conclusion and Future Directions} \label{sec:conclusion}
We introduce PUSH, an LMAPF framework that matches the scale of reactive planners like EPIBT while maintaining the benefits of long-horizon reasoning. PUSH uses subset planning and windowed planning to generate long-horizon plans for the planning queue, while utilizing an EPIBT-style recursive backtracking structure to handle planning in highly congested scenarios. This makes PUSH the first LMAPF algorithm to utilize subset planning while supporting general graph structures and to be able to coordinate up to 10,000 agents with long-horizon planning. 

\subsubsection{Future Direction 1: Long-Horizon Applications}
In this work, we demonstrate PUSH's strong empirical performance in two representative scenarios. In practice, many other MAPF applications remain challenging due to the need for simultaneous long-horizon reasoning and large-fleet coordination. We hope our PUSH framework opens a compelling pathway toward solving these challenges. For example, real-world robots are subject to kinodynamic constraints, such as robot velocity/acceleration limits. Successful planners in this setting must anticipate momentum and braking distances during dense multi-agent interactions. Existing approaches rely on windowed planning to capture these long-horizon dependencies~\cite{YanAAAI25} but scale to only a few hundred agents. Extending PUSH’s low-level A* planner to incorporate kinodynamic state spaces could potentially scale these solutions to support thousands of agents.
Similarly, MAPF under strict deadlines or delivery windows remains bottlenecked by joint windowed planning overheads \cite{mapf-deadline,amapf-deadline}. Integrating task urgency directly into PUSH's dynamic priority structure would give higher priority deliveries precedence during recursive planning, paving the way for real-time deployment at scale.

\subsubsection{Future Direction 2: Algorithmic Enhancements}
While kept simple to establish a new LMAPF paradigm, PUSH offers several promising avenues for future enhancement.
First, the current planning horizon $W$ is fixed. A natural extension is to dynamically vary a $W_i$ per agent based on the local agent density.  
Second, PUSH inherits EPIBT's restriction that an agent may push at most one lower-priority agent. Recent work~\cite{Jiang2026md-pibt} extends EPIBT to support pushing multiple agents, and this capability could be incorporated into PUSH. 
Third, agents in the planning queue $Q$ are replanned sequentially. An alternative is to dynamically merge some agents into a meta-agent and replan them jointly using more optimal MAPF algorithms, which would combine the joint-optimization strengths of RHCR with the extreme scalability of PUSH. Finally, recent advances for PIBT, such as learning-based improvements \cite{scrimp,sillm} and congestion-aware guidance \cite{ggo,traffic-flow}, could be incorporated into PUSH to further boost computational speed and coordination effectiveness.

\bibliographystyle{ieeetr}
\bibliography{citation}

\appendix

\section{Full Theoretical Analysis}
We derive the two theorems stated in the main body.
\addtocounter{theorem}{-1}
\addtocounter{theorem}{-1}

\begin{theorem}
At every timestep, PUSH is guaranteed to return $\Pi$ in finite time, and $\Pi$ is guaranteed to be collision-free. 
\end{theorem}

\begin{proof}
We first prove that PUSH always terminates in finite time; the key step is to show that the recursive call to \textsc{PushPlan} terminates in finite time.
The completion of \textsc{PushPlan} depends on both its recursion and its inner loop being bounded. The depth of recursion is bounded, as each pushed agent is added to the set $\mathcal{H}$, preventing cyclical pushing. Since the number of agents is finite, this prevents an infinite recursion. The internal loop (line \ref{line:loop} in Algorithm \ref{alg:push-plan}) is repeated as each recursive call to a \textsc{PushPlan} for agent $r_d$ fails. Upon each failure, the constraint set $\textit{cons}$ grows with a new spatiotemporal constraint (cannot collide with agent $r_d$ at timestep $t_d$). Since \textsc{PlanPath} cannot violate any previously created constraint, it cannot output the same path twice. This, combined with the fact that the number of $W$-step paths for an agent is bounded, means that there is a maximum number of times that the loop can repeat before \textsc{PlanPath} returns $\pi = \emptyset$ and the loop terminates.
Thus, the loop is bounded, and $\textsc{PushPlan}$ terminates in finite time. 

Next, we prove that, after each call of \textsc{PushPlan} on line \ref{line:push_plan_call} in Algorithm \ref{alg:push}, $\Pi$ is guaranteed to be collision-free. 
Suppose that this call is made on agent $r_1$. The \textsc{PlanAgent} procedure can only return (1) a path that is collision-free, (2) a path that collides with a single agent, or (3) no solution. In case 1, $\Pi[r_1]$ is replaced with a new collision-free path, after which we return \textsc{true} (line \ref{line:ret_true}).

In case 2, \textsc{PushPlan} is recursively called on a pushed agent. Suppose that this agent is $r_2$. Before this call returns, a number of recursive calls might be made (limited to the total number of agents with lower priority than $r_1$). Suppose each successive \textsc{PushPlan} call is made to agents in consecutive order (i.e. agent $r_j$ pushes agent $r_{j+1}$, etc.). Then each call adds an agent to the protected set (line \ref{alg:no_col}), forcing agent $r_j$ to avoid collision with $\{r_1, \cdots, r_{j-1}\}$, along with the original protected set $\mathcal{H}$ of $r_1$ (defined on line \ref{line:rec_stack_set} in Algorithm \ref{alg:push}). Eventually, since the recursive calls are limited, we will reach a \textsc{PushPlan} call on an agent $r_i$ that enters case 1 or 3. If this agent enters case 1, then it must have found a collision-free path with all agents. At this point, $\Pi$ is collision-free, since each successive \textsc{PushPlan} call avoids collisions with the agents in calls prior. Thus, each \textsc{PushPlan} call returns \textsc{true}, yielding a collision-free $\Pi$. The other scenario is that \textsc{PushPlan} cannot find a path for $r_i$ and we enter case 3. In this case, the old path of $r_i$ is restored (line \ref{alg:rec_fail_start}) and $r_{i-1}$ replans. If $r_2$ fails to plan, then agent $r_1$ is forced to replan. When this happens, a constraint is added to \textit{cons} (line \ref{alg:add_avoid_set}) and we commence a new iteration of the loop (line \ref{line:loop}). However, note that $\Pi[r_1]$ is a $W$-step collision-free path. This means it cannot be blocked by \textit{cons}, since if it was blocked by a constraint $(r_d, t_d)$, that would imply $r_d$ and $r_1$'s path intersect at time $t_d$, which is not possible for a collision-free path. Thus, $\Pi[r_1]$ cannot be blocked, and \textsc{PlanAgent} is guaranteed to return it in some loop iteration. Since this path satisfies case 1, it will always succeed. Since this does not change any of the paths, $\Pi$ remains collision-free. This also proves that case 3 never occurs for agent $r_1$, namely the agent that \textsc{PushPlan} tries to replan for on line \ref{line:push_plan_call} in Algorithm \ref{alg:push}.

Since each call of \textsc{PushPlan} on line \ref{line:push_plan_call} returns with collision-free $\Pi$ in finite time, the for loop in Algorithm \ref{alg:push} (line \ref{line:for-loop}) will terminate with collision-free $\Pi$ in finite time.
\end{proof}

Building on the above proof, we can derive the time complexity of PUSH. We characterize its time complexity by bounding the number of invocations of \textsc{PlanPath}, namely the A* search used to compute a $W$-step path for a single agent.

\begin{prop}
At each timestep, PUSH (Algorithm \ref{alg:push}) invokes the \textsc{PlanPath} procedure (i.e., an A* search) at most $N^2MP$ times, where $N$ is the number of agents, $M$ is the maximum revisit limit, $W$ is the planning window size, and $P$ is the number of possible $W$-step paths for an agent.
\end{prop}

\begin{proof}
    The loop of \textsc{PushPlan} (Algorithm \ref{alg:push-plan} line \ref{line:loop}) can be repeated a maximum of $P$ times, once for each potential $W$-step path. The \textsc{PushPlan} algorithm itself can recursively call itself and proceed to this loop a maximum of $M$ times for each of the $N$ agents, since \textit{visitCnt} restricts the number of times \textsc{PushPlan} can be called per agent to $M$ times (line \ref{alg:max_visit}). Thus, the total number of \textsc{PlanPath} calls in \textsc{PushPlan} is $NMP$. Since \textit{visitCnt} is reset to $0$ for all agents before every \textsc{PushPlan} call in Algorithm \ref{alg:push} (line \ref{line:visit_count_set}), this bound remains the same for all of these calls. If every agent is in the planning queue $Q$ in Algorithm \ref{alg:push}, then \textsc{PushPlan} is called once for each, thus the total \textsc{PlanPath} calls of PUSH is $N^2MP$.
\end{proof}

Next, we prove that PUSH preserves the deadlock-free guarantee of EPIBT, namely all agents can eventually reach their goals under the same assumptions as EPIBT. We accomplish this by showing that the highest-priority agent can always make progress towards its goal. Then we will show that this leads to all agents reaching their goals.

\begin{figure*}[t]
\centering
\includegraphics[width=0.95\textwidth]{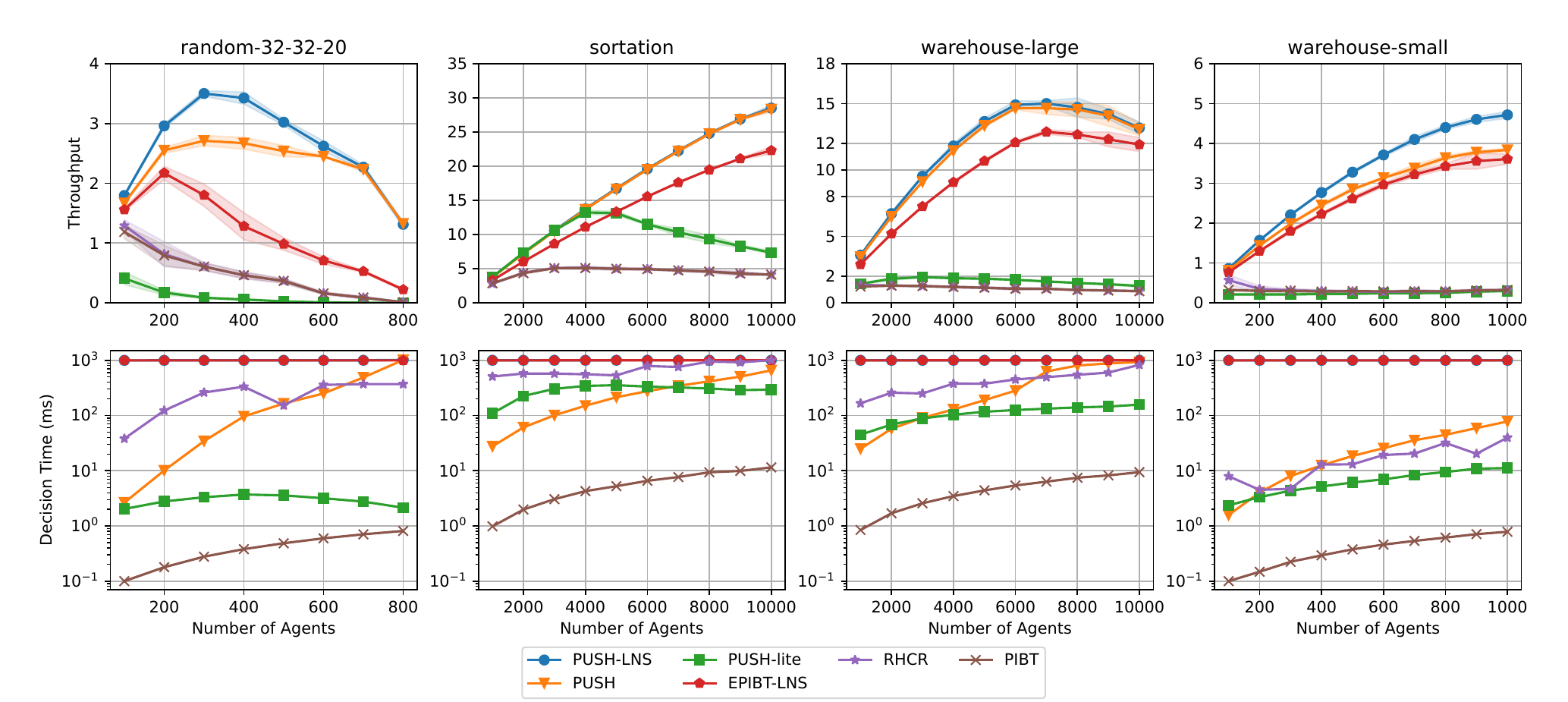} 
\caption{Comparison of PUSH with baselines on the rotation motion model, assuming a task completion time of 20 timesteps.}
\label{fig:rotation}
\end{figure*}

\begin{lemma}
    If the graph is biconnected and TCT is 0, then, whenever the highest-priority agent, denoted by $r_1$, is replanned, PUSH guarantees to find a new collision-free path $\pi$ for agent $r_1$ that strictly decreases its distance to its goal, i.e., $h(\pi[W], g_1) < h(\Pi[r_1][W], g_1)$, where $h$ is the shortest path distance function and $\Pi[r_1]$ refers to the path of $r_1$ after line \ref{line:do_action} of Algorithm \ref{alg:push}.
\end{lemma}
\begin{proof}
    Since the planning queue $Q$ is sorted by priority (line \ref{line:sort_priority} of Algorithm \ref{alg:push}), the initial call to $\textsc{PushPlan}$ replans the path for agent $r_1$. Since $r_1$ is the highest-priority agent, the set of higher-priority agents is empty. Furthermore, since TCT is 0, the set of agents executing tasks is also empty. Therefore, the protected set $\mathcal{H}$ is empty. 
    
    To facilitate the proof, we define $X_1$ as the set of $W$-step paths $\pi$ for agent $r_1$ satisfying (1) $h(\pi[W], g_1) < h(\Pi[r_1][W], g_1)$, and (2) $\pi$ collides with either no agents in $\Pi$ (excluding itself) or exactly one agent at step $W$. Note that this set is always nonempty because we can construct a path $\pi_1^* \in X_1$ by replicating $\Pi[r_1]$ for $W - 1$ steps and changing the last step to be one step closer to its goal. Since the last step of $\Pi[r_1]$ is an appended wait action (line \ref{line:do_action}) while $\pi_1^*$ moves one step closer it $g_1$, $\pi_1^*$ satisfies condition 1. Furthermore, since $\pi_1^*$ only deviates from $\Pi[r_1]$ at step $W$, it can collide with a maximum of one agent at step $W$, satisfying condition 2. 
    
    To prove the lemma, we aim to show that \textsc{PlanAgent} (called on line \ref{alg:call_pap} of Algorithm \ref{alg:push-plan}) either finds a path better than any path in $X_1$ (i.e. is closer to the goal at step $W$) and returns true, or finds a path in $X_1$, and this path will always return \textsc{true}. The former case clearly satisfies the lemma, since all paths in $X_1$ already strictly decrease their distance to the goal at step $W$. The rest of this proof focuses on the latter case, where we must prove (1) \textsc{PushPlan} cannot return \textsc{false} before trying at least one path in $X_1$ and (2) this path will result in a return of \textsc{true}.

    We start by proving that \textsc{PlanAgent} always returns a path in $X_1$ in one of the loop iterations in \textsc{PushPlan} before \textsc{PushPlan} returns \textsc{false}. Since $\mathcal{H}$ is empty, the only possibility for \textsc{PlanAgent} to not return any paths in $X_1$ is that all of them are blocked by constraints in \textit{cons}. However, we can show that for \textit{cons} to block a path in $\pi \in X_1$, a different path $\pi' \in X_1$ had to already have been returned previously. To show this fact, suppose that $\pi$ is indeed blocked by \textit{cons}. Since $\pi$ is collision-free with other agents for the first $W - 1$ steps and PUSH only adds constraints at vertices/edges at timesteps that lead to collisions with other agents, it could only have been blocked by a constraint at step $W$. 
    However, the existence of this constraint implies that in a previous loop iteration there was a path $\pi'$ that collided with an agent only at step $W$. $\pi'$ would also satisfy $h(\pi'[W], g_1) \leq h(\pi[W], g_1)$ since \textsc{PlanAgent} returns paths that take $r_1$ closer to its goal first. Therefore, $\pi' \in X_1$, and our statement holds.

    Next, we show that finding any $\pi \in X_1$ from \textsc{PlanAgent} will result in a return of \textsc{true} for this loop iteration. We can accomplish this in a similar way to how PIBT \cite{pibt} does the same: by exploiting a cycle in the biconnected graph. If $\pi$ collides with no agents, then PUSH simply returns \textsc{true}. Otherwise, suppose $\pi$ collides with an agent $r_2$ at step $W$. Agent $r_2$ inherits the priority of agent $r_1$ and is replanned with $\mathcal{H} = \{r_1\}$. Similarly to $\pi^*_1$, we can construct $\pi^*_2$, which replicates $\Pi[r_2]$ for $W - 1$ steps but makes a move at step $W$ to avoid $r_1$'s push (but not necessarily closer to its goal). Agent $r_2$ can always pivot to a vertex that does not cause a collision with agent $r_1$ at step $W$ due to the graph biconnectivity. This is because biconnectivity implies we can construct a cycle starting with the last two positions in agent $r_1$'s path, $\mathbf{C} = (\pi[W-1], \pi[W], \cdots)$, guaranteeing that $\pi^*_2$ can move from $\pi[W]$ to the next vertex in the cycle without colliding with agent $r_1$ at steps $W - 1$ and $W$ respectively. This can be extended to any agent $r_i$ on the cycle at timestep $W$, for which we can define $\pi_i^*[W]$ as the path that replicates $\Pi[r_i]$ for $W - 1$ steps and moves to the next vertex in the cycle on step $W$. Similarly to $X_1$, we can define $X_i$ to be set of paths for agent $r_i$ that do not collide with any agent for $W - 1$ steps, and on step $W$, moves to the next vertex in the cycle. For similar reasons to $r_1$, a  \textsc{PushPlan} call for $r_i$ cannot return \textsc{false} without trying at least one path in $X_i$ (such as $\pi^*_i$). Now, we must prove that $r_2$ using any path in $X_2$ will return \textsc{true} and allow $r_1$'s path to remain. Suppose for the sake of contradiction that the \textsc{PushPlan} call for agent $r_2$ returns \textsc{false}. Since $r_2$ always tries a path in $X_2$, this implies that agent $r_2$ attempted to push an agent $r_i \neq r_1$ positioned at $\pi^*_2[W]$ on step $W$ and that push returned false. This is also true for $r_i$, the agent it pushes, and so on, which would include all agents on the cycle $\mathbf{C}$. However, the existence of this cycle directly implies that some agent $r_d$ on the cycle at step $W$ must exist with a free neighbor. Even if the cycle was full, $r_d$ would be the agent at the cycle position behind $\pi[W - 1]$, allowing it to shift to $r_1$'s old position. Thus, this call actually returns \textsc{true}, so the \textsc{PushPlan} call for agent $r_2$ must return \textsc{true}. Therefore, agent $r_1$ can also return \textsc{true}.
\end{proof}

\begin{theorem}
    If the longer elapsed time is used as the priority ordering, the graph is biconnected, and TCT is 0, then PUSH is deadlock-free: every agent reaches its goal within a finite number of timesteps.
\end{theorem}

\begin{proof}
    At timestep 0, assume that agent $r_1$ is the highest-priority agent. Since agent $r_1$ remains the highest-priority agent until it reaches its goal, lemma 1 applies across all timesteps that $r_1$ is moving to its goal. Thus, $h(\Pi[r_1][W], g_1)$ is strictly decreasing whenever $r_1$ replans. Since agent $r_1$ is always replanned at latest every $W$ steps, $\Pi[r_1][W]$ will reach $g_1$ in finite timesteps. After $\Pi[r_1][W]$ reaches $g_1$, the subsequent \textsc{PushPlan} calls will never take a path for agent $r_1$ that reaches $g_1$ later than its previous path. Thus, $r_1$ is guaranteed to reach its goal. Once $r_1$ reaches its goal, the second-highest-priority agent will become the highest priority. For the same reasoning, it will also reach its goal in a finite amount of time. As each agent becomes the highest priority, they will all be able to eventually reach their goals.
\end{proof}

In cases when $T > 0$, we cannot provide the same guarantee. The key difference in this case is that the highest-priority agent may attempt to push an agent that is executing its task at step $W$. Since this agent cannot be pushed, the rest of the proof cannot follow.
An example case where this might happen is if the highest-priority agent is trapped between two agents who continuously receive tasks at their current locations. In this case, these two agents will stay in their locations forever, trapping the highest-priority agent between them and rendering it unable to reach its goal. This counterexample also holds for PIBT and EPIBT.

\section{Additional Experiments}
We describe additional experiments and ablation studies. In these experiments, all hyperparameters remain the same as in the main results, unless otherwise stated.

\subsection{Rotation Motion Model}
\Cref{fig:rotation} shows the results. PUSH-LNS and PUSH remain the highest maximum throughput across all maps, demonstrating that longer horizon times also improve performance in the rotation model. This model challenges reactive planners because agents need multiple moves to turn and vacate a cell. While EPIBT-LNS still handles the rotation model well, PIBT and RHCR perform poorly in every map. While PIBT's one-step decision making process cannot easily push agents aside, RHCR fails to scale well as agent count and density increase.

\end{document}